\documentclass[]{interact}

\usepackage{epstopdf}
\usepackage[caption=false]{subfig}

\usepackage[numbers,sort&compress]{natbib}

\bibpunct[, ]{[}{]}{,}{n}{,}{,}

\theoremstyle{plain}
\newtheorem{theorem}{Theorem}[section]

\newtheorem{proposition}[theorem]{Proposition}

\theoremstyle{definition}
\newtheorem{definition}[theorem]{Definition}

\theoremstyle{remark}

\begin{document}


\title{Regret theory under fear of the unknown}

\author{
\name{Fang Liu 
\thanks{Corresponding author. f\_liu@gxu.edu.cn; fang272@126.com}}
\affil{
School of Mathematics and Information Science, Guangxi University, \\Nanning Guangxi 530004, China;
}
}

\maketitle

\begin{abstract}
It is common to encounter the situation with uncertainty for decision makers (DMs) in dealing with a complex decision making problem. The existing evidence shows that people usually fear the extreme uncertainty named as the unknown. This paper reports the modified version of the typical regret theory by considering the fear experienced by DMs for the unknown.
Based on the responses of undergraduate students to the hypothetical choice problems with an unknown outcome, some experimental evidences are observed and analyzed. The framework of the modified regret theory is established by considering the effects of an unknown outcome. A fear function is equipped and some implications are proved. The behavioral foundation of the modified regret theory is further developed by modifying the axiomatic properties of the existing one as those based on the utility function; and it is recalled as the utility-based behavioral foundation for convenience. The application to the medical decision making with an unknown risk is studied and the effects of the fear function are investigated. The observations reveal that the existence of an unknown outcome could enhance, impede or reverse the preference relation of people in a choice problem, which can be predicted by the developed regret theory.
\end{abstract}

\begin{keywords}
Regret theory, Fear of the unknown, Experimental evidence, Fear function, Behavioral foundation
\end{keywords}

\section{Introduction}\label{sec:intro}

Regret theory offers an effective framework for choosing two prospects with state-contingent consequences by capturing the experienced emotions of DMs in the past \citep{Bell1982,Loomes1982,Diecidue2017}. It seems common to experience the negative sensation of regret when a foregone prospect exhibits a better outcome than the chosen one.
The psychology of regret aversion plays a great role in a future decision making, which has been verified according to the psychological content and neurobiological evidence \citep{Camille2004,Bour2010}. It is worth noting that the theoretical and practical investigations related to regret theory have attracted a great deal of attention \citep{Loomes1987,Fishburn1989,Quiggin1994,Bleichrodt2010,Bleichrodt2015,Diecidue2017}. The recent findings have shown that the typical regret theory is based on the solid measurement theory and the behavioral foundation \citep{Bleichrodt2010,Diecidue2017}.

It is convenient to recall the formalization process of the typical regret theory. A finite state space is assumed as $S=\{s_{1}, s_{2}, \cdots, s_{n}\}$ by following the existing ideas \citep{Loomes1982,Diecidue2017}.
The probability of each state $s_{i}$ is written as $p_{i}\in[0,1]$ with
$
\sum_{i=1}^{n}p_{i}=1.
$
When a real-valued outcome $x_{i}$ is considered for the occurrence of the state $s_{i}$  $(i\in I_{n}=\{1,2,\cdots,n\}),$ a prospect can be determined as $X=(x_{1}, p_{1}; x_{2}, p_{2}; \cdots; x_{n}, p_{n})$ \citep{Kahneman1979}.
The next important issue is how to choose between two independent prospects such as $f=(f_{1}, p_{1}; f_{2}, p_{2}; \cdots; f_{n}, p_{n})$ and
$g=(g_{1}, p_{1}; g_{2}, p_{2}; \cdots; g_{n}, p_{n}).$ We further suppose that for any individual, there is a choiceless utility function $u(\cdot)$ by following the utility theory \citep{Neumann1953,Fishburn1970}, which is continuous and strictly increasing. Under the framework of regret theory, when choosing $f$ and rejecting $g,$ the DM should possess the utility $u(f_{i})$ and lose the utility $u(g_{i})$ with the occurrence of the $i$th state.
Simultaneously, it is considered that the final utility of the DM could be affected by the difference between $u(f_{i})$ and $u(g_{i}).$ When $u(f_{i})>u(g_{i}),$ the final utility could be bigger than $u(f_{i})$ due to the experienced rejoicing.
When $u(f_{i})<u(g_{i}),$ the final utility could be less than $u(f_{i})$ due to the experienced regret. Hence, a strictly increasing function $R(\cdot)$ with three times differentiable is assumed under the condition $R(0)=0.$ For the case of choosing $f$ and rejecting $g,$ the modified utility of the DM corresponding to the occurrence of the $i$th state is expressed as:
\begin{equation}
F_{i}(f,g)=u(f_{i})+R(u(f_{i})-u(g_{i})),\quad i\in I_{n}.\label{req1}
\end{equation}
Similar to (\ref{req1}), when choosing $g$ and rejecting $f,$ the final utility of the DM with the occurrence of the $i$th state can be written as:
\begin{equation}
F_{i}(g,f)=u(g_{i})+R(u(g_{i})-u(f_{i})),\quad i\in I_{n}.\label{req1a}
\end{equation}
In addition, by considering the following relation:
\begin{equation}
F_{i}(f,g)-F_{i}(g,f)=u(f_{i})-u(g_{i})+R(u(f_{i})-u(g_{i}))-R(u(g_{i})-u(f_{i})),
\end{equation}
for $i\in I_{n},$ the regret function $Q(\cdot)$ is defined such that $Q(\xi)=\xi+R(\xi)-R(-\xi)$ satisfying $Q(\xi)=-Q(-\xi).$ Finally, the choice between $f$ and $g$ is based on the following rule:
\begin{equation}
f\succeq g\Leftrightarrow \sum_{i=1}^{n}p_{i}Q(u(f_{i})-u(g_{i}))\geq0,\label{req3}
\end{equation}
where the symbol $\succeq$ denotes the commonly used weak preference relation. It is seen from (\ref{req3}) that the regret function $Q(\cdot)$ together with $R(\cdot)$ play an important role in the regret theory. A convex function $Q(\cdot)$ corresponds to the case with regret aversion and a linear one means the degenerative case without sensitivity to regret.

It is found from the formalization process of the typical regret theory, the outcomes $x_{i}$ $(i\in I_{n})$ are always assumed to be known in the prospect $X=(x_{1}, p_{1}; x_{2}, p_{2}; \cdots; x_{n}, p_{n}).$ However, when dealing with a complex decision making problem, one could face the situation that some of the outcomes corresponding to the states are unknown.
For example, when a woman decides to get pregnant, many aspects of pregnancy could be unknown and out of the control \citep{Jones2005}. The unknown outcomes could play a different role in the final utility of the DM as compared to the known ones. In other words, when there is an unknown outcome for a state, the utility function $u(\cdot)$ should be reconstructed. The underlying reason is attributed to the fact that people have a natural fear of the unknown, which has been illustrated carefully in the application and review papers \citep{Cao2011,Carleton2016}.
It is clear that the effects of fear of the unknown on the utility of the DM have not been considered in the typical regret theory.
In order to deal with the important situation, we carry out the experimental study within the undergraduate students by a series of hypothetical choice problems in this study.
A comparative analysis is made for the cases with and without an unknown outcome for a state, respectively. A modified choice rule is further proposed by introducing the fear function and a novel behavioral foundation is established.

The structure of the paper is organized as follows. In Section 2, the experimental evidences are presented to reveal that the utility function $u(\cdot)$ should be reconsidered due to the existence of the unknown outcome for a state. By comparing the responses of university students about the hypothetical choice problems with and without an unknown outcome, the typical evidences for regret theory are observed again. It is further found that the existence of an unknown outcome could enhance, impede or reverse the preference relation under the case without an unknown outcome. The new experimental evidences under the case of an unknown outcome show that the typical regret theory should be improved. Section 3 offers the formalization process of the modified regret theory by considering the new experimental evidences. The effect of an unknown outcome on the utility function is incorporated into the modified regret theory by proposing a fear function. The implications of the modified regret theory are further investigated. One can find that the novel experimental phenomenon of choosing prospects can be predicted by using the developed regret theory.
Section 4 reports a theoretical analysis for the behavioral foundation of the modified regret theory. Following the idea in \cite{Diecidue2017}, it is observed that the behavioral foundation of the modified regret theory can be achieved by adjusting the construction of the utility function in the typical regret theory and forming the axioms according to the modified utility function. In Section 5, the medical decision problem is analyzed by considering the existence of an unknown outcome. The combined effects of the fear functions and the probability of the unknown outcomes are discussed. The predicted results in terms of the modified regret theory show that an unknown risk existing in a medical decision making could affect greatly the choice of people. The main conclusions are covered in Section 6.

\section{Experimental evidences}

\begin{table}[ptb]
 \caption{Responses to hypothetical choice problems in \cite{Kahneman1979} and the present study.}
\resizebox{\textwidth}{!}{
\begin{tabular}
[c]{c|c|c|c|c|c}
 \hline
Cases &\multicolumn{2}{c|}{\underline{The offered prospects$^{\dag}$}}&Modal preference&PSMP$^{*}$(KT)&PSMP$^{*}$(P)\\
&$f$&$g$&&\\
 \hline
1&(2500, 0.33; 2400, 0.66; 0, 0.01)&(2400, 1.00)&$f\prec g$&$82\%$&$52\%$\\
2&(2500, 0.33; 0, 0.67)&(2400, 0.34; 0, 0.66)&$f\succ g$&$83\%$&$67\%$\\
3&(4000, 0.8; 0, 0.2)&(3000, 1.00)&$f\prec g$&$80\%$&$56\%$\\
4&(--4000, 0.8; 0, 0.2)&(--3000, 1.00)&$f\succ g$&$92\%$&$63\%$\\
5&(4000, 0.2; 0, 0.8)&(3000, 0.25; 0, 0.75)&$f\succ g$&$65\%$&$70\%$\\
6&(--4000, 0.2; 0, 0.8)&(--3000, 0.25; 0, 0.75)&$f\prec g$&$58\%$&$68\%$\\
7&(6000, 0.45; 0, 0.55)&(3000, 0.90; 0, 0.10)&$f\prec g$&$86\%$&$81\%$\\
8&(6000, 0.001; 0, 0.999)&(3000, 0.002; 0, 0.998)&$f\succ g$&$73\%$&$79\%$\\
9&(5000, 0.001; 0, 0.999)&(5, 1.00)&$f\succ g$&$72\%$&$53\%$\\
10&(--5000, 0.001; 0, 0.999)&(--5, 1.00)&$f\prec g$&$83\%$&$56\%$\\
\hline
\end{tabular}}\\
$\dag$ The wealth measured in Israeli Pounds or RMB.
$*$ Statistically significant at the 0.01 level.
\label{lab1a}
\end{table}

In the following, we analyze the responses of university students to some hypothetical choice problems.
First, it is of much interest to recall the existing experimental results shown in Table \ref{lab1a}, where the term of PSMP denotes the percentage of subjects with modal preference \citep{Kahneman1979,Loomes1982}. It is seen that the wealth is measured in Israeli Pounds and the statistically significant level is 0.01. The amounts of the offered outcomes are significant since the median net monthly income for an Israeli family is about 3,000 Israeli pounds in 1970s \citep{Kahneman1979}.
In the present study, the same choice problems in Table \ref{lab1a} were offered to the undergraduate students at Guangxi University in China by replacing Israeli Pounds with RMB. It should be pointed out that the median monthly
cost of living for an undergraduate student at Guangxi University is about 1500 RMB. The involved amounts of wealth in Table \ref{lab1a} are still significant to the undergraduate students at Guangxi University. The participants were asked to carefully consider the choice between the two prospects $f$ and $g$ by imagining the actual situations.
In the instructions, it was specified that there was no "correct" answer and the responses should follow their hearts by considering the existing risk.
It is found that the modal preferences are identical to those in Table \ref{lab1a} according to the received 109 responses. The difference is the percentage of subjects with modal preference and the observed results are given in the term of PSMP$^{*}$(P). Comparing the numbers in PSMP$^{*}$(KT) and PSMP$^{*}$(P), it seems that the present participants are more adventurous than those in 1970s.

\begin{table}[ptb]
 \caption{Responses to hypothetical choice problems with an unknown outcome.}
\resizebox{\textwidth}{!}{
\begin{tabular}
[c]{c|c|c|c|c|c}
 \hline
Cases&\multicolumn{2}{c|}{\underline{The offered prospects$^{\dag}$}}&Modal&PSMP$^{*}$&Effect of\\
&$\bar{f}$&$\bar{g}$&preference&&the unknown\\
 \hline
1&(2500, 0.33; 2400, 0.66; unknown, 0.01)&(2400, 1.00)&$\bar{f}\prec \bar{g}$&$54\%$&Increasing\\
2&(2500, 0.33; unknown, 0.67)&(2400, 0.34; 0, 0.66)&$\bar{f}\prec \bar{g}$&$55\%$&Reversal\\
2$^\prime$&(2500, 0.33; 0, 0.67)&(2400, 0.34; unknown, 0.66)&$\bar{f}\succ \bar{g}$&$63\%$&Decreasing\\
3&(4000, 0.8; unknown, 0.2)&(3000, 1.00)&$\bar{f}\prec \bar{g}$&$58\%$&Increasing\\
4&(--4000, 0.8; unknown, 0.2)&(--3000, 1.00)&$\bar{f}\succ \bar{g}$&$66\%$&Increasing\\
5&(4000, 0.2; unknown, 0.8)&(3000, 0.25; 0, 0.75)&$\bar{f}\succ \bar{g}$&$56\%$&Decreasing\\
5$^\prime$&(4000, 0.2; 0, 0.8)&(3000, 0.25; unknown, 0.75)&$\bar{f}\prec \bar{g}$&$52\%$&Reversal\\
6&(--4000, 0.2; unknown, 0.8)&(--3000, 0.25; 0, 0.75)&$\bar{f}\prec \bar{g}$&$63\%$&Decreasing\\
6$^\prime$&(--4000, 0.2; 0, 0.8)&(--3000, 0.25; unknown, 0.75)&$\bar{f}\succ \bar{g}$&$54\%$&Reversal\\
7&(6000, 0.45; unknown, 0.55)&(3000, 0.90; unknown, 0.10)&$\bar{f}\prec \bar{g}$&$64\%$&Decreasing\\
8&(6000, 0.001; unknown, 0.999)&(3000, 0.002; unknown, 0.998)&$\bar{f}\succ \bar{g}$&$84\%$&Increasing\\
9&(5000, 0.001; unknown, 0.999)&(5, 1.00)&$\bar{f}\succ \bar{g}$&$51\%$&Decreasing\\
10&(--5000, 0.001; unknown, 0.999)&(--5, 1.00)&$\bar{f}\prec \bar{g}$&$67\%$&Increasing\\
\hline
\end{tabular}}\\
$\dag$ The wealth measured in RMB.
$*$ Statistically significant at the 0.01 level.
\label{lab2a}
\end{table}

Moreover, it is seen from Table \ref{lab1a} that the outcomes in the provided prospects are all known. The complexity of the practical problem may lead to the case that there is an unknown outcome.  The 109 participants were invited again to make a choice between the two prospects $\bar{f}$ and $\bar{g}$ with an unknown outcome. The effects of the unknown outcomes on the modal preference and the PSMP are reported in Table \ref{lab2a}.
In the following, let us compare the observations in Tables \ref{lab1a} and \ref{lab2a}, and the three situations are discussed.
\begin{itemize}
\item[$\bullet$] \emph{Preference reversal.}

The term ``Reversal'' in Table \ref{lab2a} means that the modal preference is reversed by comparing the corresponding case in Table \ref{lab1a}.
For example, the difference between the two prospects $f=(2500, 0.33; 0, 0.67)$ and $\bar{f}=(2500, 0.33; unknown, 0.67)$ is that the zero outcome is replaced by the unknown one. By comparing with $g=\bar{g}=(2400, 0.34; 0, 0.66),$ the modal preference of the participants is reversed from $f\succ g$ with the percentage $67\%$ to $\bar{f}\prec \bar{g}$ with the percentage $55\%.$
This means that the unknown outcome could change the preference relation of the DMs; and the changing trend is with the unknown averse. The above phenomenon is in agreement with the
behaviour of people to fear the unknown \citep{Carleton2016}.

\item[$\bullet$] \emph{Percentage increasing.}

It is seen from Table \ref{lab2a} that the term ``Increasing'' implies the increasing percentage of modal preference with the preference preservation.
As an illustration, we consider the cases of $3$ and $4$ in Tables \ref{lab1a} and \ref{lab2a}, respectively.
When the prospect $f=(4000, 0.8; 0, 0.2)$ in Table \ref{lab1a} is replaced by the prospect $\bar{f}=(4000, 0.8; unknown, 0.2)$ in Table \ref{lab2a}, the modal preference
$f\prec g$ is not changed and the percentage is increasing from $56\%$ to $58\%$ by comparing $g=\bar{g}=(3000, 1.0).$ The observation still reflects the unknown averse under the fear psychology of the unknown.
Moreover, when the prospect $f=(-4000, 0.8; 0, 0.2)$ is rewritten as $\bar{f}=(-4000, 0.8; unknown, 0.2),$ the modal preference $f\succ g$ holds with an increasing percentage
$66\%$ by comparing $g=\bar{g}=(-3000, 1.0).$ This means that people do not fear the unknown for this situation and exhibit the unknown loving.

\item[$\bullet$] \emph{Percentage decreasing.}

When the percentage of modal preference is decreasing and without the preference reversal, it is marked as ``Decreasing'' in Table \ref{lab2a}.
For instance, the cases $5$ and $6$ in Tables \ref{lab1a} and \ref{lab2a} are used to analyze the effects of the unknown outcome.
When replacing the prospect $f=(4000, 0.2; 0, 0.8)$ in Table \ref{lab1a} by $\bar{f}=(4000, 0.2; unknown, 0.8)$ and keeping $g=\bar{g}=(3000, 0.25; 0, 0.75),$ the percentage of $f\succ g$ decreases from $70\%$ to $56\%$ of $\bar{f}\succ \bar{g}.$
This phenomenon reveals the aversion of the unknown risk due to the existence of the unknown outcome. In addition,
when replacing the prospect $f=(-4000, 0.2; 0, 0.8)$ in Table \ref{lab1a} by $\bar{f}=(-4000, 0.2; unknown, 0.8)$ and letting $g=\bar{g}=(-3000, 0.25; 0, 0.75),$ there is a $5\%$ decrease for the percentage of $f\prec g$ from $68\%$ to $63\%$ of $\bar{f}\prec \bar{g}.$ This implies that people have some loving for the unknown outcome.

\end{itemize}

The above analysis shows that the existence of the unknown outcome has a great influence on the choice of the prospects. By comparing the several situations, it is found that the effects of the unknown outcome on the model preference are dependent on the probability of the occurrence of the unknown outcome and the signs of the known outcome. Hence, the regret theory is worth to be further developed by considering the effects of the unknown outcome.

\section{The modified regret theory and its implications}

Now we construct the theoretical framework of the modified regret theory by considering the existence of an unknown outcome. The basic notations and the choosing rule are proposed; then the experimental observations are illustrated by proving the implications.

\subsection{The framework of the modified regret theory}

For convenience, a finite state space $S=\{s_{1}, s_{2}, \cdots, s_{n}\}$ is still assumed
by following the basic ideas of regret theory \citep{Bell1982,Loomes1982}. We can further integrate all states with an unknown outcome as a particular one $s_{u}.$
Then the finite state space is assumed to be  $\bar{S}=\{s_{1}, s_{2}, \cdots, s_{m}; s_{u}\}$ with $1\leq m\leq n.$
The corresponding probability of each state $s_{i}$ is expressed as $p_{i}\in[0,1]$ for $i\in I_{m}=\{1,2,\cdots,m\}$ or $i=u$ with the following condition:
\begin{equation}
\sum_{i=1}^{m}p_{i}+p_{u}=1.
\end{equation}
One can see that when $p_{u}=0,$ the state space $\bar{S}$ degenerates to that in the typical regret theory. According to the state space $\bar{S},$
a prospect is expressed as $\bar{X}=(x_{1}, p_{1}; x_{2}, p_{2}; \cdots; x_{m}, p_{m}; x_{u}, p_{u}),$ where the symbol $x_{u}$ stands for the unknown outcome.
Similar to the consideration in \cite{Loomes1982}, it is assumed that there exists a choiceless utility function $u(\cdot)$ such that the obtained utility can be measured as a value of real numbers $u(x_{i})$ for the occurrence of the state $s_{i}$ $(i\in I_{m});$ and $u(x_{i})$ is continuous and strictly increasing. However, there is a natural difference between the determined outcome $x_{i}$ $(i\in I_{m})$ and the unknown one $x_{u}.$ The psychological experience of fear of the unknown is dependent on the probability of the occurrence of the unknown event. For example, when we face a completely unknown world, a great sense of fear could be experienced. When one encounters an unknown thing with zero probability, a relaxed feeling should be experienced. Therefore, it is supposed that there is a normalized fear function
 $v_{i}(p_{u})$ such that the following conditions are satisfied:
\begin{equation}
v_{i}(0)=1,\quad v_{i}(1)=0,\quad i\in I_{m}.\label{req6a}
\end{equation}
Moreover, the fear function $v(p_{u})$ can be also considered to be continuous and strictly decreasing.
With the existence of $v(p_{u}),$ the utility function $u(x_{i})$ should be modified and it is written as:
\begin{equation}
\bar{u}(x_{i})=v_{i}(p_{u})u(x_{i}).\label{req7a}
\end{equation}
This means that the existence of an unknown outcome has the influence on the utility of the determined outcomes.
The idea in the formula (\ref{req7a}) is similar to the consideration in \cite{Bell1981} that ``people will pay a premium to avoid uncertainty in financial returns."
In particular, we should answer the question that what is the utility value of the unknown outcome. It is noted that some existing works always focus on the utility value of an outcome with an imprecise probability \citep{Du2005}. For the unknown outcome, it is assumed that the utility value is given as:
\begin{equation}
\bar{u}(x_{u})=0.\label{eq8a}
\end{equation}
The assumption (\ref{eq8a}) is in accordance with the case with the condition of $v_{i}(1)=0$ in (\ref{req7a}). It means that when one is faced with a complete unknown situation, the utility value is zero due to the extreme fear.

In the following, let us compare any two independent prospects:
\begin{eqnarray}
&&\bar{f}=(f_{1}, p_{1}; f_{2}, p_{2}; \cdots; f_{n}, p_{n}), \nonumber\\
&&\bar{g}=(g_{1}, p_{1}; g_{2}, p_{2}; \cdots; g_{n}, p_{n}).\nonumber
\end{eqnarray}
By considering the existence of unknown outcomes, it is assumed that the sums of the probabilities corresponding to the unknown outcomes in $\bar{f}$ and $\bar{g}$ are $p_{fu}$ and $p_{gu},$ respectively.
By using the regret function $R(\cdot),$ when choosing $\bar{f}$ and rejecting $\bar{g},$ the utility value of the DM for the occurrence of the state $s_{i}$ is expressed as:
\begin{equation}
\bar{F}_{i}(\bar{f},\bar{g})=\bar{u}(f_{i})+R(\bar{u}(f_{i})-\bar{u}(g_{i})),\quad i\in I_{n}.\label{req8a}
\end{equation}
where
$$
\bar{u}(f_{i})=v_{i}(p_{fu})u(f_{i}),\quad \bar{u}(g_{i})=v_{i}(p_{gu})u(g_{i}).
$$
When choosing $\bar{g}$ and rejecting $\bar{f},$ we have
\begin{equation}
\bar{F}_{i}(\bar{g},\bar{f})=\bar{u}(g_{i})+R(\bar{u}(g_{i})-\bar{u}(f_{i})),\quad i\in I.\label{req9a}
\end{equation}
The regret function $Q(\cdot)$ is further defined as:
\begin{equation}
Q(\xi)=\xi+R(\xi)-R(-\xi).
\end{equation}
The choice mechanism of $\bar{f}$ and $\bar{g}$ is established as:
\begin{equation}
\bar{f}\succeq \bar{g}\Leftrightarrow \Psi=\sum_{i=1}^{n}p_{i}Q(\bar{u}(f_{i})-\bar{u}(g_{i}))\geq0.\label{req11a}
\end{equation}

Comparing with the typical choice rule of regret theory (\ref{req3}), the expression (\ref{req11a}) considers the effects of an unknown outcome according to the modified utility function (\ref{req7a}).
With the knowledge of (\ref{req6a}), when $p_{fu}=p_{gu}=0,$ the choice mechanism with (\ref{req11a}) degenerates to (\ref{req7a}). When $p_{fu}=p_{gu}=1,$ we have $\bar{f}\sim \bar{g}$ due to the zero utility with $Q(\bar{u}(f_{i})-\bar{u}(g_{i}))=0$ for $\bar{u}(f_{i})=\bar{u}(g_{i})=0$ $(i\in I_{n}).$

\subsection{The implications of the modified regret theory}

In what follows, we focus on the implications of the modified regret theory to the choices between two independent prospects with an unknown outcome.

(1) {\bf The ratio effect and its reverse}

  The following result can be predicted by using the modified regret theory:
  \begin{proposition}\label{pro1}
  Let $\bar{f}=(f_{1}, \lambda p; unknown, 1-\lambda p)$ and $\bar{g}=(g_{1}, p; 0, 1- p)$ be two independent prospects with $p\in(0,1]$ and $\lambda\in(0,1).$
  If there is a probability $p=\bar{p}$ such that $\bar{f}\sim \bar{g},$ then (I) $p>\bar{p}\Rightarrow \bar{f}\succ \bar{g}$ and $p<\bar{p}\Rightarrow \bar{f}\prec \bar{g}$
  for $f_{1}>g_{1}>0$ and $0<v_{1}(p_{u})u(f_{1})<u(g_{1});$ (II) $p>\bar{p}\Rightarrow \bar{f}\prec \bar{g}$ and $p<\bar{p}\Rightarrow \bar{f}\succ\bar{g}$
  for $f_{1}<g_{1}<0$ and $u(g_{1})<v_{1}(p_{u})u(f_{1})<0.$
  \end{proposition}
\begin{proof}
  In order to prove Proposition \ref{pro1}, we follow the idea in \cite{Loomes1982} to construct the decision matrix in Table \ref{tab:m} for the two independent prospects $\bar{f}=(f_{1}, p_{f}; unknown, 1-p_{f})$ and $\bar{g}=(g_{1}, p_{g}; 0, 1-p_{g}).$ It is seen from Table \ref{tab:m} that the outcomes for the states $s_{3}$ and $s_{4}$ are unknown when choosing $\bar{f}$ and rejecting $\bar{g}.$ In the proposed theoretical framework, the utility value is zero with respect to the unknowns, and the unknown outcomes have influences on the utility values with respect to $s_{1}$ and $s_{2}.$

\begin{table}[ptb]
\caption{The decision matrix for $\bar{f}=(f_{1}, p_{f}; unknown, 1-p_{f})$ and $\bar{g}=(g_{1}, p_{g}; 0, 1-p_{g}).$}
\tabcolsep 0.05in
\begin{tabular}
[c]{c|cccc}
 \hline
States &$s_{1}$&$s_{2}$&$s_{3}$&$s_{4}$\\
\hline
Probability&$p_{f}p_{g}$&$p_{f}(1-p_{g})$&$p_{g}(1-p_{f})$&$(1-p_{f})(1-p_{g})$\\
\hline
Choosing $\bar{f}$ and rejecting  $\bar{g}$ &$f_{1}$&$f_{1}$&unknown&unknown\\
Choosing $\bar{g}$ and rejecting $\bar{f}$ &$g_{1}$&0 &$g_{1}$&0 \\
\hline
\end{tabular}\label{tab:m}
\end{table}

  In addition, it is convenient to choose $u(0)=0$ \citep{Loomes1982}. According to (\ref{req11a}), we have
\begin{eqnarray}
\bar{f}\succeq \bar{g}&\Leftrightarrow &p_{f}Q(v_{1}(p_{fu})u(f_{1}))-p_{g}Q(u(g_{1}))+p_{f}p_{g}\left[Q(v_{1}(p_{fu})u(f_{1})-u(g_{1}))\right.\nonumber\\
&&-\left.Q(v_{1}(p_{fu})u(f_{1}))+Q(u(g_{1}))\right]\geq0,\label{req12a}
\end{eqnarray}
where $p_{fu}=1-p_{f}.$
Then by considering $p_{f}=\lambda p$ and $p_{g}=p,$ it follows:
\begin{eqnarray}
\bar{f}\succeq\bar{g}\Leftrightarrow \Psi(p, \lambda, p_{fu}, f_{1}, g_{1})\geq0,\label{req13a}
\end{eqnarray}
where
\begin{eqnarray}
\Psi(p, \lambda, p_{fu}, f_{1}, g_{1})&=&p\left\{\lambda Q(v_{1}(p_{fu})u(f_{1}))- Q(u(g_{1}))\right.\nonumber\\
&&+\lambda p[Q(v_{1}(p_{fu})u(f_{1})-u(g_{1}))\nonumber\\
&&-\left.Q(v_{1}(p_{fu})u(f_{1}))+Q(u(g_{1}))]\right\}.
\end{eqnarray}
Under the assumption of a convex function $Q(\xi)$ for $\forall \xi>0,$ when $f_{1}>g_{1}>0$ and $0<v_{1}(p_{fu})u(f_{1})<u(g_{1}),$ it gives the following inequality:
\begin{equation}
Q(v_{1}(p_{fu})u(f_{1})-u(g_{1}))>Q(v_{1}(p_{fu})u(f_{1}))-Q(u(g_{1})).
\end{equation}
Therefore, for the case of $p>\bar{p},$ we have $v_{1}(1-p)>v_{1}(1-\bar{p})$ and
$$Q(v_{1}(1-p)u(f_{1}))>Q(v_{1}(1-\bar{p})u(f_{1})).$$
The above conditions yield
$\Psi(p, \lambda, p_{fu}, f_{1}, g_{1})>0$
and $\bar{f}\succ \bar{g}.$ For the case of $0<p<\bar{p},$ one has the relations of $v_{1}(1-p)<v_{1}(1-\bar{p})$ and $$Q(v_{1}(1-p)u(f_{1}))<Q(v_{1}(1-\bar{p})u(f_{1})).$$
Then it follows
$\Psi(p, \lambda, p_{fu}, f_{1}, g_{1})<0$
and $\bar{f}\prec \bar{g}.$

The above proof procedure shows that the ratio effect is satisfied under the condition of $0<v_{1}(p_{fu})u(f_{1})<u(g_{1}).$ On the contrary, when $f_{1}<g_{1}<0$ and $u(g_{1})<v_{1}(p_{fu})u(f_{1})<0,$  the reverse ratio effect can be achieved, which is due to the following inequalities:
\begin{eqnarray}
&&Q(v_{1}(p_{fu})u(f_{1})-u(g_{1}))>Q(v_{1}(p_{fu})u(f_{1}))-Q(u(g_{1})),\nonumber\\
&&Q(v_{1}(1-p)u(f_{1}))>Q(v_{1}(1-\bar{p})u(f_{1})),\quad for \quad 0<p<\bar{p},\nonumber\\
&&Q(v_{1}(1-p)u(f_{1}))<Q(v_{1}(1-\bar{p})u(f_{1})),\quad for \quad p>\bar{p}.\nonumber
\end{eqnarray}
The proof is completed.
\end{proof}

Let us further analyze some experimental evidences in Table \ref{lab2a}. For example,
we choose $\bar{f}_{1}=(4000, 0.80; unknown, 0.20)$ and $\bar{g}_{1}=(3000, 1.0)$ with the result of $\bar{f}_{1}\prec \bar{g}_{1}.$ In terms of Proposition \ref{pro1}, we have $p=1.0,$ $p_{u}=0.2,$ and $\lambda=0.8,$ meaning that $\bar{p}>0.2$ under $0<v_{1}(p_{fu})u(4000)<u(3000).$ Moreover, when the case with $\bar{f}_{1}=(4000, 0.2; unknown, 0.8)$ and $\bar{g}_{1}=(3000,0.25; 0, 0.75)$ is considered,
one arrives at $\bar{f}_{1}\succ\bar{g}_{1},$ $p=0.25,$ $p_{u}=0.8$ and $\lambda=0.8,$ implying that $\bar{p}<0.25.$ In other words, when $0<v_{1}(p_{fu})u(4000)<u(3000),$
there exists a probability $0.2<\bar{p}<0.25$ such that the experimental observations can be predicted by using the modified regret theory.
As compared to the common ratio effect and its inverse shown in \cite{Loomes1982}, the novelty is the introduction of the factor $v_{1}(p_{u})$ in Proposition \ref{pro1}. The fear of the unknown yields the reverse relation of $0<v_{1}(p_{fu})u(f_{1})<u(g_{1})$ for $f_{1}>g_{1}>0,$ or $0>v_{1}(p_{fu})u(f_{1})>u(g_{1})$ for $f_{1}<g_{1}<0.$

(2) {\bf The preference reversal effect}

  The existence of the unknown outcome may lead to the preference reversal of the choice; and the following prediction can be reached.

  \begin{proposition}\label{pro2}
  Let $f=(f_{1}, \lambda p; 0, 1-\lambda p)$ and $g=(g_{1}, p; 0, 1-p)$ be two independent prospects with $p\in(0,1],$  $\lambda\in(0,1)$ and $p_{fu}=1-\lambda p.$
  (I) Under $f\succ g$ for $f_{1}>g_{1}>0,$ if $f=(f_{1}, \lambda p; 0, 1-\lambda p)$ is changed to $\bar{f}=(f_{1}, \lambda p; unknown, 1-\lambda p)$ and the condition of $0< v_{1}(p_{fu})u(f_{1})<u(g_{1})$ is satisfied, then there is a sufficiently small $p>0$ such that the preference relation $\bar{f}\prec\bar{g}$ holds.
  (II) Under $f\prec g$ for $f_{1}<g_{1}<0,$ if $f=(f_{1}, \lambda p; 0, 1-\lambda p)$ is changed to $\bar{f}=(f_{1}, \lambda p; unknown, 1-\lambda p)$ and the condition of $ u(g_{1})<v_{1}(p_{fu})u(f_{1})<0$ is satisfied, then there is a sufficiently small $p>0$ such that the preference relation $\bar{f}\succ\bar{g}$ holds.
  \end{proposition}
\begin{proof}

(I) Based on the typical regret theory, the condition of $f\succ g$ for $f_{1}>g_{1}>0$ leads to $\Phi(p, \lambda, f_{1}, g_{1})>0$ where
\begin{eqnarray}
 \Phi(p, \lambda, f_{1}, g_{1})&=& p\left\{\lambda Q(u(f_{1}))- Q(u(g_{1}))\right.\nonumber\\
 &&+\left.\lambda p\left[Q(u(f_{1})-u(g_{1}))-Q(u(f_{1}))+Q(u(g_{1}))\right]\right\},\label{eqf16}
\end{eqnarray}
with $p<\bar{p}$ and $\Phi(\bar{p}, \lambda, f_{1}, g_{1})=0.$
When there is an unknown outcome, we follow the proof procedure of Proposition \ref{pro1} and give the choice rule in (\ref{req13a}).
Moreover, since the function $Q(\xi)$ is convex and $Q(-\xi)=-Q(\xi)$ for $\xi>0,$ the application of $0<v_{1}(p_{fu})u(f_{1})<u(g_{1})$ and $f_{1}>g_{1}>0$ leads to
\begin{eqnarray}
&&Q(v_{1}(p_{fu})u(f_{1})-u(g_{1}))>Q(v_{1}(p_{fu})u(f_{1}))-Q(u(g_{1})),\nonumber\\
&&Q(u(f_{1})-u(g_{1}))<Q(u(f_{1}))-Q(u(g_{1})).\nonumber
\end{eqnarray}

The above two inequalities mean the following relation:
 \begin{eqnarray}
&&Q(v_{1}(p_{fu})u(f_{1})-u(g_{1}))-Q(v_{1}(p_{fu})u(f_{1}))\nonumber\\
&&-Q(u(f_{1})-u(g_{1}))+Q(u(f_{1}))>0.
\end{eqnarray}
By comparing (\ref{req13a}) and (\ref{eqf16}), we have
\begin{eqnarray}
&&\Psi(p, \lambda, p_{u}, f_{1}, g_{1})-\Phi(p, \lambda, f_{1}, g_{1})\nonumber\\
&=&p\lambda Q(v_{1}(p_{fu})u(f_{1}))-p\lambda Q(u(f_{1}))\nonumber\\
&&+p^{2}\lambda[Q(v_{1}(p_{fu})u(f_{1})-u(g_{1}))-Q(v_{1}(p_{fu})u(f_{1}))\nonumber\\
&&-Q(u(f_{1})-u(g_{1}))+Q(u(f_{1}))]\nonumber\\
&<&p\lambda Q(v_{1}(p_{fu})u(f_{1}))-p\lambda Q(u(f_{1}))\nonumber\\
&&+p\lambda[Q(v_{1}(p_{fu})u(f_{1})-u(g_{1}))-Q(v_{1}(p_{fu})u(f_{1}))\nonumber\\
&&-Q(u(f_{1})-u(g_{1}))+Q(u(f_{1}))]\nonumber\\
&=&-p\lambda\left[Q(u(g_{1})-v_{1}(p_{fu})u(f_{1}))+Q(u(f_{1})-u(g_{1}))\right]<0.
\end{eqnarray}
This means that with the decreasing $p,$ the value of $\Psi(p, \lambda, p_{fu}, f_{1}, g_{1})$ is decreasing and less than $\Phi(p, \lambda, f_{1}, g_{1}).$ In addition, for the particular case with $p\rightarrow0$ and $p_{fu}\rightarrow1,$ it gives $v_{1}(p_{fu})\rightarrow0$ and
\begin{equation}
\Psi(p, \lambda, p_{fu}, f_{1}, g_{1})\rightarrow-pQ(u(g_{1}))<0.
\end{equation}
This implies that there exists a sufficiently small $p>0$ such that
$$\Psi(p, \lambda, p_{fu}, f_{1}, g_{1})<0,$$
 and the preference relation $\bar{f}\prec\bar{g}$ holds.

(II) According to the condition of $f\prec g$ for $f_{1}<g_{1}<0,$ one has $\Phi(p, \lambda, f_{1}, g_{1})<0$ for $p<\bar{p}$ with $\Phi(\bar{p}, \lambda, f_{1}, g_{1})=0.$
The application of $u(g_{1})<v_{1}(p_{fu})u(f_{1})<0$ and $f_{1}<g_{1}<0$ yields the following relations:
\begin{eqnarray}
&&Q(v_{1}(p_{fu})u(f_{1})-u(g_{1}))<Q(v_{1}(p_{fu})u(f_{1}))-Q(u(g_{1})),\nonumber\\
&&Q(u(f_{1})-u(g_{1}))>Q(u(f_{1}))-Q(u(g_{1})).\nonumber
\end{eqnarray}
Then it follows:
 \begin{eqnarray}
&&Q(v_{1}(p_{fu})u(f_{1})-u(g_{1}))-Q(v_{1}(p_{fu})u(f_{1}))\nonumber\\
&&-Q(u(f_{1})-u(g_{1}))+Q(u(f_{1}))<0.
\end{eqnarray}
The difference between $\Psi(p, \lambda, p_{fu}, f_{1}, g_{1})$ and $\Phi(p, \lambda, f_{1}, g_{1})$ can be computed as follows:
\begin{eqnarray}
&&\Phi(p, \lambda, f_{1}, g_{1})-\Psi(p, \lambda, p_{fu}, f_{1}, g_{1})\nonumber\\
&=&p\lambda Q(u(f_{1}))-p\lambda Q(v_{1}(p_{fu})u(f_{1}))\nonumber\\
&&-p^{2}\lambda[Q(v_{1}(p_{fu})u(f_{1})-u(g_{1}))-Q(v_{1}(p_{fu})u(f_{1}))\nonumber\\
&&-Q(u(f_{1})-u(g_{1}))+Q(u(f_{1}))]\nonumber\\
&<&p\lambda Q(u(f_{1}))-p\lambda Q(v_{1}(p_{fu})u(f_{1}))\nonumber\\
&&-p\lambda[Q(v_{1}(p_{fu})u(f_{1})-u(g_{1}))-Q(v_{1}(p_{fu})u(f_{1}))\nonumber\\
&&-Q(u(f_{1})-u(g_{1}))+Q(u(f_{1}))]\nonumber\\
&=&p\lambda\left[Q(u(f_{1})-u(g_{1}))-Q(v_{1}(p_{fu})u(f_{1})-u(g_{1}))\right]<0.
\end{eqnarray}
This means that the value of $\Psi(p, \lambda, p_{fu}, f_{1}, g_{1})$ is increasing and larger than $\Phi(p, \lambda, f_{1}, g_{1})$ due to the existence of $v_{1}(p_{fu}).$
When $p\rightarrow0$ and $p_{fu}\rightarrow1,$ it gives $v_{1}(p_{fu})\rightarrow0$ and
\begin{equation}
\Psi(p, \lambda, p_{fu}, f_{1}, g_{1})\rightarrow-pQ(u(g_{1}))>0.
\end{equation}
This means that there is a sufficiently small $p>0$ such that $\Psi(p, \lambda, p_{fu}, f_{1}, g_{1})>0,$ and
the preference relation $\bar{f}\succ\bar{g}$ holds.

\end{proof}

For instance, we analyze the experimental observations in Table \ref{lab2a}. Let us assume that $f=(2500, 0.33; 0, 0.67),$ $\bar{f}=(2500, 0.33; unknown, 0.67)$
and $g=\bar{g}=(2400, 0.34; 0, 0.66).$ Then it gives $p=0.34,$  $\lambda=33/34$ and $p_{fu}=0.67;$
the experimental result shows $f\succ g.$ When the unknown outcome is introduced, we have $\bar{f}\prec \bar{g}$ for $p=0.34.$

(3) {\bf The reflection effect of the bilinear utility function}

 In the above analysis, the utility function $\bar{u}(x, p_{fu})=v_{1}(p_{fu})u(x)$ is monotonically increasing with respect to $x$ and monotonically decreasing with respect to $p_{fu}.$
Now we consider the simplest case that the utility function is bilinear, meaning that the utility function $\bar{u}(x, p_{fu})=v_{1}(p_{fu})u(x)$ is linear with respect $p_{fu}$ and $x$ such as $\bar{u}(x, p_{fu})=(1-p_{fu})x.$
Let $\bar{f}=(f_{1}, p_{f}; unknown, 1-p_{f})$ and $\bar{g}=(g_{1}, p_{g}; 0, 1- p_{g})$ be two independent prospects with $p\in(0,1]$ and $\lambda\in(0,1).$
According to (\ref{req11a}), we have
\begin{equation}
\bar{f}\succeq \bar{g}\Leftrightarrow p_{f}Q(p_{f}f_{1})-p_{g}Q(g_{1})+
p_{f}p_{g}\left[Q(p_{f}f_{1}-g_{1})-Q(p_{f}f_{1})+Q(g_{1})\right]\geq0.\label{eq24}
\end{equation}
Moreover, we suppose that $\bar{f}^{\prime}=(-f_{1}, p_{f}; unknown, 1-p_{f})$ and $\bar{g}^{\prime}=(-g_{1}, p_{g}; 0, 1-p_{g}).$ It follows:
\begin{equation}
\bar{f}^{\prime}\preceq \bar{g}^{\prime}\Leftrightarrow p_{f}Q(p_{f}f_{1})-p_{g}Q(g_{1})+
p_{f}p_{g}\left[Q(p_{f}f_{1}-g_{1})-Q(p_{f}f_{1})+Q(g_{1})\right]\geq0.\label{eq25}
\end{equation}
The relations in (\ref{eq24}) and (\ref{eq25}) show that
\begin{equation}
\bar{f}\succeq \bar{g}\Leftrightarrow\bar{f}^{\prime}\preceq \bar{g}^{\prime}.
\end{equation}
This means that the reflection effect can be observed under the bilinear assumption of the utility function similar to the finding in \cite{Loomes1982}.
As shown in Table \ref{lab2a}, the experimental evidence exhibits the reflection effect, for example, $(4000, 0.8; unknown, 0.2)=\bar{f}\prec\bar{g}=(3000, 1.00)$ and
$(-4000, 0.8; unknown, 0.2)=\bar{f}^{\prime}\succ\bar{g}^{\prime}=(-3000, 1.00).$

The above analysis mainly focuses on the case with $\bar{f}=(f_{1}, p_{f}; unknown, 1-p_{f})$ and $\bar{g}=(g_{1}, p_{g}; 0, 1- p_{g}).$ The other cases can be studied similarly when
$\bar{f}=(f_{1}, p_{f}; unknown, 1-p_{f})$ and $\bar{g}=(g_{1}, p_{g}; unknown, 1- p_{g});$ or $\bar{f}=(f_{1}, p_{f}; 0, 1-p_{f})$ and $\bar{g}=(g_{1}, p_{g}; unknown, 1- p_{g}).$
Moreover, one can see that here the choice between the two independent prospects is only investigated. It is clear that the choice among multiple prospects is worth to be investigated in the future.

\section{The behavioral foundation}

It is seen from the framework and implications of the modified regret theory that the experimental evidences can be predicted.
The preference foundation of the modified regret theory should be investigated similar to that of the typical regret theory \citep{Diecidue2017}. Here we first recall the behavioral foundation of regret theory provided by \cite{Diecidue2017}; then it is developed to establish the behavioral foundation of the modified regret theory under fear of the unknown.

\subsection{The behavioral foundation of the typical regret theory}

For convenience, the formal definition of the typical regret theory is recalled as follows \citep{Loomes1982,Diecidue2017}:
\begin{definition}\label{def1}
Regret theory holds when the relation (\ref{req3}) is satisfied for a strictly increasing continuous utility function $u,$ subjective probabilities $p_{i},$ and a strictly increasing, skew symmetric continuous regret function $Q.$
\end{definition}
It is further assumed that $\mathbb{R}^{n}$ stands for the set of all prospects represented by $f=(f_{1}, f_{2}, \cdots, f_{n})$ under a finite state space $S=\{s_{1}, s_{2}, \cdots, s_{n}\}.$
The symbol $\alpha_{i}f$ denotes the prospect $f$ where the outcome $f_{i}$ is replaced by $\alpha.$ When the probability of the occurrence of the state $s_{i}$ is considered, the corresponding prospect is expressed as $\alpha_{p_{i}}f.$ Then the axiomatic properties about the preference relation $\succeq$ are given as follows \citep{Diecidue2017}:
\begin{itemize}
  \item[A1.] Completeness: $f\succeq g$ or $g\succeq f$ for $\forall f, g\in \mathbb{R}^{n};$
  \item[A2.] Transitivity: $f\succeq g, g\succeq h\Rightarrow f\succeq h$ for $f,g,h\in\mathbb{R}^{n};$
  \item[A3.] Weak monotonicity: For $\forall f, g\in\mathbb{R}^{n},$ if $f_{i}\geq g_{i}$ for $i\in I,$ then $f\succeq g;$
  \item[A4.] Strong monotonicity: For $\forall f, g\in\mathbb{R}^{n},$ $f\succ g$ if $f_{i}\geq g_{i}$ for $i\in I,$ and there is a non-null state $s_{i}$ such that $f_{i}>g_{i}.$
  A null state $s_{i}$ means $\alpha_{i}f\sim \beta_{i}f$ for $\forall f\in \mathbb{R}^{n}$ and the outcomes $\alpha$ and $\beta.$
  \item[A5.] Continuity: The subsets of $\{f\in \mathbb{R}^{n}| f\succeq g\}$ and $\{f\in \mathbb{R}^{n}| f\preceq g\}$ are closed for each $g\in\mathbb{R}^{n}.$
\end{itemize}

Moreover, if there are two prospects $f$ and $g$ satisfying the following relations:
\begin{equation}
\alpha_{i}f\sim \beta_{i}g\quad and\quad \gamma_{i}f\sim \delta_{i} g,
\end{equation}
for a non-null state $s_{i}$ and the outcomes $\alpha, \beta, \gamma, \delta,$ it implies the preference relation $\sim_{t}$ with
$\alpha\beta\sim_{t}\gamma\delta.$
Therefore, the trade-off consistency is defined as follows \citep{Kobberling2004}:
\begin{definition}\label{def2}
For the prospects $f,g,x,y\in\mathbb{R}^{n},$ the trade-off consistency is satisfied when the conditions of $\alpha_{i}f\sim\beta_{i}g, \gamma_{i}f\sim\delta_{i}g$ and $\alpha_{i}x\sim\beta_{i}y$ imply $\gamma_{i}x\sim\delta_{i}y.$
\end{definition}
In addition, the notation $\succeq_{t}$ with $\alpha\beta\succeq_{t}\gamma\delta$ is defined if there
are two prospects $f,g$ and a non-null state $s_{i}$ such that
\begin{equation}
\alpha_{i}f\succeq_{t}\beta_{i}g,\quad and\quad \gamma_{i}f\preceq_{t}\delta_{i} g.
\end{equation}
The notation $\succ_{t}$ with $\alpha\beta\succ_{t}\gamma\delta$ means that at least  one of
$\alpha_{i}f\succeq_{t}\beta_{i}g$ and $\gamma_{i}f\preceq_{t}\delta_{i} g$ is strict.
Hence, the definition of the preference trade-off consistency is offered as follows \citep{Wakker1984}:
\begin{definition}\label{def3}
For the prospects $f,g,x,y\in\mathbb{R}^{n},$ the preference trade-off consistency holds when there do not exist four outcomes $\alpha,\beta,\gamma,\delta$ such that
$\alpha\beta\succeq_{t}\gamma\delta$ and $\gamma\delta\succ_{t}\alpha\beta.$
\end{definition}
The strict dominance with $f\succeq_{SD}g$ indicates that we have $f_{i}\geq g_{i}$ for $i\in I$ and $f_{i}>g_{i}$ for a non-null state $s_{i}.$
Then the definition of the d-transitivity is given as follows:
\begin{definition}\label{def4}
For $f,g,h\in\mathbb{R}^{n},$ the d-transitivity holds if the conditions of $f\succeq_{SD}g$ and $g\succeq h$ mean $f\succ h$ and the conditions
of $f\succeq g$ and $g\succeq_{SD} h$ imply $f\succ h.$
\end{definition}

At the end, the behavioral foundation of the typical regret theory is characterized by the following result \citep{Diecidue2017}:
\begin{theorem} \label{th1}
Regret theory holds if and only if the preference relation $\succeq$ satisfies the five properties with completeness,
d-transitivity, strong monotonicity, continuity and trade-off consistency.
\end{theorem}
The finding in Theorem \ref{th1} shows the sufficient and necessary conditions of the regret theory. It has been pointed out in  \cite{Diecidue2017} that the properties of the preference relation $\succeq$ have a small difference as compared to those of the utility theory. The key issue is that the transitivity in the utility theory is replaced by the d-transitivity in Theorem \ref{th1}.

\subsection{The utility-based behavioral foundation of the modified regret theory}

In the following, let us construct the behavioral foundation of the modified regret theory by considering the fear of the unknown. Without loss of generality, we assume that $v_{1}(x)=v_{2}(x)=\cdots=v_{n}(x)=v(x).$ Then the formal definition of the modified regret theory is given as follows:
\begin{definition}\label{def5}
Regret theory under fear of the unknown holds when the relation (\ref{req11a}) is satisfied for a strictly increasing continuous utility function $u: \mathbb{R}\mapsto\mathbb{R},$ a strictly decreasing continuous fear function $v: [0, 1]\mapsto[0,1]$ subject to $v(0)=1$ and $v(1)=0,$ subjective probabilities $p_{i},$ and a strictly increasing, skew symmetric continuous regret function $Q: \mathbb{R}\mapsto\mathbb{R}.$
\end{definition}
Comparing with Definition \ref{def1}, the main difference in Definition \ref{def5} is that there exists a fear function $v(x)$ such that the effect of an unknown outcome on the utility can be considered.
Due to the existence of an unknown outcome, the notations should be rewritten in order to construct the behavioral foundation.
For a finite state space $S=\{s_{1}, s_{2}, \cdots, s_{n}\},$ the state-contingent prospect $f$ is expressed as $f=(f_{1}, p_{1}; f_{2}, p_{2}; \cdots; f_{n}, p_{n})$ and simplified as
$f=(f_{1}, f_{2}, \cdots, f_{n}).$
When there is an unknown outcome such as $f_{m},$ the set of the prospects is written as $\mathbb{R}^{n}_{u}.$
It is noted that the utility function $u(\cdot)$ is strictly increasing continuous with $u: \mathbb{R}\mapsto\mathbb{R}.$ Obviously, the value of the utility
for an unknown outcome is not defined in the existing works. Here we consider that the unknown outcome should lead to the unknown psychological experience of pleasure.
The utility value of an unknown outcome $f_{m}$ is determined as $u(f_{m})=0,$ which can be understood from the mean value of utility about the positive and negative outcomes.
Moreover, the effects of an unknown outcome on the utility values of the known outcomes should be considered. The sum of the probabilities of all unknown outcomes in a prospect $f$ is written as $p_{u}$ and the fear function is expressed as $v(p_{u}).$ Therefore, the final utility of an outcome $f_{i}$ $(i\in I)$ can be obtained as
\begin{equation}
\bar{u}(f_{i})=v(p_{u})u(f_{i}),\quad i\in I.
\end{equation}
Hence, the axiomatic properties about the preference relation $\succeq$ are adjusted as follows:
\begin{itemize}
  \item[A$^\prime$1.] Completeness: $\bar{f}\succeq \bar{g}$ or $\bar{g}\succeq \bar{f}$ for $\forall \bar{f}, \bar{g}\in \mathbb{R}^{n}_{u};$
  \item[A$^\prime$2.] Transitivity: $\bar{f}\succeq \bar{g}, \bar{g}\succeq \bar{h}\Rightarrow \bar{f}\succeq \bar{h}$ for $\forall \bar{f}, \bar{g}, \bar{h}\in\mathbb{R}^{n}_{u};$
  \item[A$^\prime$3.] Weak monotonicity: For $\forall \bar{f}, \bar{g}\in\mathbb{R}^{n}_{u},$ if $\bar{u}(f_{i})\geq \bar{u}(g_{i})$ for $i\in I,$ then $\bar{f}\succeq \bar{g};$
  \item[A$^\prime$4.] Strong monotonicity: For $\forall \bar{f}, \bar{g}\in\mathbb{R}^{n}_{u},$ $\bar{f}\succ \bar{g}$ if $\bar{u}(f_{i})\geq \bar{u}(g_{i})$ for $i\in I,$ and there is a non-null state $s_{i}$ such that $\bar{u}(f_{i})>\bar{u}(g_{i}).$
  \item[A$^\prime$5.] Continuity: The subsets of $\{\bar{f}\in \mathbb{R}^{n}_{u}| \bar{f}\succeq \bar{g}\}$ and $\{\bar{f}\in \mathbb{R}^{n}_{u}| \bar{f}\preceq \bar{g}\}$ are closed for each $\bar{g}\in\mathbb{R}^{n}_{u}.$
\end{itemize}
It is seen that the preference relations $\bar{f}\succeq \bar{g}$ and $\bar{f}\succ \bar{g}$ in A$^\prime$3 and A$^\prime$4 are based on the utility values of $\bar{u}(f_{i})$ and $\bar{u}(g_{i})$ for $\forall i\in I.$ As shown in A1-A5, the preference relations are constructed by using the values of the outcomes. Clearly, the axiomatic properties in A$^\prime$1-A$^\prime$5 are in agreement with the utility theory \citep{Neumann1953,Fishburn1970}.

In addition, similar to Definitions \ref{def2} and \ref{def3}, the trade-off consistency and preference trade-off consistency can be defined on $\mathbb{R}^{n}_{u}.$ For example, we have
the following definition:
\begin{definition}\label{def2a}
The trade-off consistency holds for the prospects $\bar{f}, \bar{g}, \bar{x}, \bar{y}\in\mathbb{R}^{n}_{u},$ if the application of $\alpha_{i}\bar{f}\sim\beta_{i}\bar{g}, \gamma_{i}\bar{f}\sim\delta_{i}\bar{g}$ and $\alpha_{i}\bar{x}\sim\beta_{i}\bar{y}$ leads to $\gamma_{i}\bar{x}\sim\delta_{i}\bar{y}.$
\end{definition}
For giving the definition of the d-transitivity, we redefine the strict dominance $\succeq_{SD}$ as follows:
\begin{definition}\label{def6}
The relation of $\bar{f}\succeq_{SD}\bar{g}$ is satisfied if $\bar{u}(f_{i})\geq\bar{u}(g_{i})$ for all $i,$ and there is a non-null state $s_{i}$ such that $\bar{u}(f_{i})>\bar{u}(g_{i}).$
\end{definition}
Based on Definition \ref{def6}, the d-transitivity can be given similar to Definition \ref{def4}:
\begin{definition}\label{def4a}
The d-transitivity is satisfied for $\bar{f}, \bar{g}, \bar{h}\in\mathbb{R}^{n}_{u},$ if $\bar{f}\succeq_{SD}\bar{g}$ and $\bar{g}\succeq \bar{h}$ mean $\bar{f}\succ \bar{h};$ and
 $\bar{f}\succeq \bar{g}$ and $\bar{g}\succeq_{SD} \bar{h}$ yield $\bar{f}\succ \bar{h}.$
\end{definition}
As shown in the above definitions, here the main difference to Definitions \ref{def2}-\ref{def4} is that the relations of $f_{i}\geq g_{i}$ and $f_{i}>g_{i}$ are replaced by $\bar{u}(f_{i})\geq \bar{u}(g_{i})$ and $\bar{u}(f_{i})>\bar{u}(g_{i}),$ respectively.
When all the outcomes in $\bar{f}$ and $\bar{g}$ are known, the relation of $\bar{u}(f_{i})\geq \bar{u}(g_{i})$ can degenerate to $f_{i}\geq g_{i}$ by considering the
strictly increasing property of the utility function $u(\cdot).$ In fact, when we compare two prospects $\bar{f}$ and $\bar{g},$ the utility values play the key role \citep{Neumann1953,Fishburn1970}.
Therefore, it is natural to replace $f_{i}\geq g_{i}$ and $f_{i}>g_{i}$ by the conditions of $\bar{u}(f_{i})\geq \bar{u}(g_{i})$ and $\bar{u}(f_{i})>\bar{u}(g_{i}),$ respectively.

Based on the above axiomatic properties of the preference relation $\succeq,$ the representation theorem for the modified regret theory is given as follows:
\begin{theorem} \label{th2}
Regret theory under fear of the unknown holds if and only if the preference relation $\succeq$ in (\ref{req11a}) satisfies the five modified properties of completeness,
d-transitivity, strong monotonicity, continuity and trade-off consistency.
\end{theorem}
\begin{proof}
On the one hand, when regret theory under fear of the unknown holds, we derive the five modified properties of the preference relation $\succeq$ in (\ref{req11a}).
\begin{enumerate}
  \item[(1)] Completeness: By considering any two prospects $\bar{f}=(f_{1}, f_{2}, \cdots, f_{n})$ and $\bar{g}=(g_{1}, g_{2}, \cdots, g_{n})$ in $\mathbb{R}^{n}_{u},$
  when the modified regret theory is satisfied, the values of $\Psi=\sum_{i=1}^{n}p_{i}Q(\bar{u}(f_{i})-\bar{u}(g_{i}))$ are used to evaluate the importance of $\bar{f}$ and $\bar{g}.$
  The relations of $\Psi>0, \Psi=0$ and $\Psi<0$ mean $\bar{f}\succ\bar{g}, \bar{f}\sim\bar{g}$ and $\bar{f}\prec\bar{g},$ respectively. This implies that the completeness of the preference relation $\succeq$ is satisfied.
  \item[(2)] D-transitivity: Assume that three prospects in $\mathbb{R}^{n}_{u}$ are expressed as $\bar{f}=(f_{1}, f_{2}, \cdots, f_{n}),$ $\bar{g}=(g_{1}, g_{2}, \cdots, g_{n})$ and $\bar{h}=(h_{1}, h_{2}, \cdots, h_{n}).$ When $\bar{f}\succeq_{SD}\bar{g},$ one has $\bar{u}(f_{i})\geq\bar{u}(g_{i})$ for all $i,$ and there is a non-null state $s_{i}$ such that $\bar{u}(f_{i})>\bar{u}(g_{i}).$ Since the function $Q(\cdot)$ is strictly increasing, we arrive at $\Psi_{1}=\sum_{i=1}^{n}p_{i}Q(\bar{u}(f_{i})-\bar{u}(g_{i}))>0.$ When $\bar{g}\succeq \bar{h},$ it follows $\Psi_{2}=\sum_{i=1}^{n}p_{i}Q(\bar{u}(g_{i})-\bar{u}(h_{i}))\geq0.$
      Therefore, we have
      \begin{eqnarray}
      \Psi_{3}&=&\sum_{i=1}^{n}p_{i}Q(\bar{u}(f_{i})-\bar{u}(h_{i}))\nonumber\\
      &=&\sum_{i=1}^{n}p_{i}Q(\bar{u}(f_{i})-\bar{u}(g_{i})+\bar{u}(g_{i})-\bar{u}(h_{i}))\nonumber\\
      &\geq&\sum_{i=1}^{n}p_{i}Q(\bar{u}(f_{i})-\bar{u}(g_{i}))+\sum_{i=1}^{n}p_{i}Q(\bar{u}(g_{i})-\bar{u}(h_{i}))>0,\nonumber
      \end{eqnarray}
      then $\bar{f}\succ \bar{h},$ where the convexity of $Q(\cdot)$ has been used.  Similarly, it gives $\bar{f}\succ \bar{h}$ by applying $\bar{f}\succeq \bar{g}$ and $\bar{g}\succeq_{SD} \bar{h}.$ This means that the d-transitivity is satisfied for the preference relation $\succeq.$
  \item[(3)]Strong monotonicity: For any two prospects $\bar{f}=(f_{1}, f_{2}, \cdots, f_{n})$ and $\bar{g}=(g_{1}, g_{2}, \cdots, g_{n})$ in $\mathbb{R}^{n}_{u},$ if $\bar{u}(f_{i})\geq \bar{u}(g_{i})$ for all $i,$ and $\exists s_{i}$ such that $\bar{u}(f_{i})>\bar{u}(g_{i}),$ $\Psi=\sum_{i=1}^{n}p_{i}Q(\bar{u}(f_{i})-\bar{u}(g_{i}))>0.$ This implies $\bar{f}\succ\bar{g}$ due to the regret theory under fear of the unknown.
  \item[(4)] Continuity: For $\forall \bar{f}, \bar{g}\in\mathbb{R}^{n}_{u},$ the continuity of $v(\cdot)$ and $u(\cdot)$ yields the continuity of $\bar{u}(\cdot).$ Then using the continuity of $Q(\cdot),$ the subsets of $\{\bar{f}\in \mathbb{R}^{n}_{u}| \bar{f}\succeq \bar{g}\}$ and $\{\bar{f}\in \mathbb{R}^{n}_{u}| \bar{f}\preceq \bar{g}\}$ are closed.
  \item[(5)] Trade-off consistency. Let four prospects $\bar{f}, \bar{g}, \bar{x}$ and $\bar{y}$ in $\mathbb{R}^{n}_{u}$ be written as $\bar{f}=(f_{1}, f_{2}, \cdots, f_{n}),$ $\bar{g}=(g_{1}, g_{2}, \cdots, g_{n}),$  $\bar{x}=(x_{1}, x_{2}, \cdots, x_{n})$ and $\bar{y}=(y_{1}, y_{2}, \cdots, y_{n}).$
  Under the modified regret theory with the outcomes $\alpha, \beta, \gamma$ and $\delta,$ the conditions of $\alpha_{i}\bar{f}\sim\beta_{i}\bar{g}$ and $\gamma_{i}\bar{f}\sim\delta_{i}\bar{g}$ lead to the following equalities:
  \begin{eqnarray}
  &&p_{i}Q(\bar{u}(\alpha)-\bar{u}(\beta))=-\sum_{k=1,k\neq i}^{n}Q(\bar{u}(f_{k})-\bar{u}(g_{k})),\nonumber\\
  &&p_{i}Q(\bar{u}(\gamma)-\bar{u}(\delta))=-\sum_{k=1,k\neq i}^{n}Q(\bar{u}(f_{k})-\bar{u}(g_{k})),\nonumber
  \end{eqnarray}
  This implies $\bar{u}(\alpha)-\bar{u}(\beta)=\bar{u}(\gamma)-\bar{u}(\delta)$ due to the strictly increasing property of $Q(\cdot).$
  Moreover, applying $\alpha_{i}\bar{x}\sim\beta_{i}\bar{y},$ we have
  \begin{equation}
  p_{i}Q(\bar{u}(\alpha)-\bar{u}(\beta))=-\sum_{k=1,k\neq i}^{n}Q(\bar{u}(x_{k})-\bar{u}(y_{k})).\nonumber
  \end{equation}
  Then it follows:
  \begin{equation}
  p_{i}Q(\bar{u}(\gamma)-\bar{u}(\delta))=-\sum_{k=1,k\neq i}^{n}Q(\bar{u}(x_{k})-\bar{u}(y_{k})).\nonumber
  \end{equation}
  This means $\gamma_{i}\bar{x}\sim\delta_{i}\bar{y}$ and the trade-off consistency is satisfied according to Definition \ref{def2a}.
\end{enumerate}

On the other hand, let us assume that the properties of completeness,
d-transitivity, strong monotonicity, continuity and trade-off consistency are satisfied for the preference relation $\succeq.$ Then we can prove that the modified regret theory under fear of the unknown holds by virtue of Definition \ref{def5}. It is feasible to follow the proof steps of regret theory shown in  \cite{Diecidue2017}.
Comparing (\ref{req3}) and (\ref{req11a}), one can find that the only difference is the construction of the utility function. In the proof procedure of the representation theorem for the typical regret theory, the utility function in (\ref{req3}) can be replaced by using a linear transformation with $u^{\prime}=a\cdot u+c$ $(a>0, \forall c\in \mathbb{R})$ \citep{Diecidue2017}. The utility function in (\ref{req11a}) is a linear transformation of $u(x)$ and it is written as $\bar{u}(x)=v(p_{u})u(x)$ $(v(p_{u})>0)$ for a fixed $p_{u}\in[0,1).$
This means that the utility function $\bar{u}(x)$ in (\ref{req11a}) can be derived by using the existing procedure in \cite{Diecidue2017}. Thus the regret theory under fear of the unknown holds when the preference relation $\succeq$ satisfies the modified properties of completeness,
d-transitivity, strong monotonicity, continuity and trade-off consistency.
\end{proof}

As compared to the finding in \cite{Diecidue2017}, one of the differences is that the introduced factor $v(p_{u})$ in the utility function $\bar{u}(x)$ exhibits a practical meaning to characterize the effect of the unknown outcomes. The other differences are attributed to the properties where the conditions of $f_{i}\geq g_{i}$ and $f_{i}>g_{i}$ are replaced by $\bar{u}(f_{i})\geq \bar{u}(g_{i})$ and $\bar{u}(f_{i})>\bar{u}(g_{i}),$ respectively. The utility theory is naturally used to characterize the axiomatic properties of the preference relations.

\section{Application to the medical decision making with an unknown outcome}

\begin{table}[ptb]
\caption{The advantages and disadvantages of two treatments.}
\resizebox{\textwidth}{!}
	{
\begin{tabular}
[c]{c|c|c}
 \hline
Two treatments &Advantages&Disadvantages\\
\hline
Surgery &Fewer side effects and &Losing voice and \\
 &less chance of recurrence&living with artificial voice\\
\hline
Radiotherapy &Normal voice&High chance of recurrence and \\
 &&less than 3 years to live if recurring\\
\hline
\end{tabular}\label{tab:d1}
}
\end{table}

\begin{table}[ptb]
\caption{The utility values for different outcomes}
\resizebox{\textwidth}{!}
	{
\begin{tabular}
[c]{c|ccc}
 \hline
Outcomes &Full health with no disease &Normal voice&Artificial voice \\
\hline
Utility values &1&0.7&0.5\\
\hline
Outcomes &Normal voice+death&Artificial voice+death&Immediate death \\
\hline
Utility values &0.28&0.27&0\\
\hline
\end{tabular}\label{tab:d2}
}
\end{table}

In the following, let us investigate the medical decision making problem with an unknown outcome.
Assume that the DM is suffered from a stage T3 laryngeal cancer without metastases \citep{Wakker2010,Diecidue2017}.
It is difficult to choose one of the surgery and radiotherapy for the stage, since there is not a clear medical instruction about the preference of the two treatments.
The advantages and disadvantages of surgery and radiotherapy are given in Table \ref{tab:d1}.
The final treatment decision is dependent on the subjective preference of the patient.
A decision analyst is invited to help the DM choosing between radiotherapy and surgery.
Suppose that the DM is with regret averse with the regret function $Q(x)=x^{3}.$ The utility values of the DM for different outcomes are shown in Table \ref{tab:d2}.
As reported in \cite{Diecidue2017}, the two independent prospects are expressed as follows:
\begin{eqnarray}
f=Surgery=(0.5, 0.6; 0.27, 0.4),\quad
g=Radiotherapy=(0.7, 0.3; 0.28, 0.7),\nonumber
\end{eqnarray}
where the outcomes in the two prospects have been replaced by the corresponding utility values.
According to the method in \cite{Loomes1982}, the decision matrix can be obtained and given in Table \ref{tab:d3}.
Based on the formula of (\ref{req3}), we can calculate
\begin{eqnarray}
&&0.18\cdot(0.5-0.7)^{3}+0.42\cdot(0.5-0.28)^{3}\nonumber\\
&&+0.12\cdot(0.27-0.7)^{3}+0.28\cdot(0.27-0.28)^{3}=-0.0065<0.\nonumber
\end{eqnarray}
This means that $f\prec g$ according to the typical regret theory when the DM exhibits regret averse.
The obtained result is not in agreement with the finding in \cite{Diecidue2017} by considering the three states \citep{Wakker2010}: (1) both radiotherapy and surgery prevent recurrence with the probability 0.3, (2) only surgery prevents recurrence with the probability 0.3, and (3) both surgery and radiotherapy cannot prevent recurrence with the probability 0.4.
Moreover, it is seen that the observation of $f\prec g$ is in accordance with the rational decision based on the behavioral foundation \citep{Diecidue2017}.
The above phenomenon shows that the application of different state spaces could lead to different choices. Here one can see that the decision matrix method in \cite{Loomes1982} is compatible with the behavioral foundation in \cite{Diecidue2017}.

\begin{table}[ptb]
\caption{The decision matrix for the medical decision problem under regret theory.}
\tabcolsep 0.20in
\begin{tabular}
[c]{c|cccc}
 \hline
States &$s_{1}$&$s_{2}$&$s_{3}$&$s_{4}$\\
\hline
Probability&0.18 &0.42&0.12&0.28\\
\hline
Utilities for surgery &0.5&0.5&0.27&0.27\\
\hline
Utilities for radiotherapy &0.7&0.28&0.7&0.28\\
\hline
\end{tabular}\label{tab:d3}

\end{table}

On the other hand, it is noted that the unknown outcomes of surgery and radiotherapy are not considered in the above analysis. For example, there must exist a certain degree of surgical risk for a surgery. Some healthy cells could be killed under radiotherapy and some unknown results could exist. We analyze the following three cases:
\begin{itemize}
\item Case I: $\bar{f}=(0.5, 0.6-0.5p_{u}; 0.27, 0.4-0.5p_{u}; \text{unknown outcome}, p_{u})$ and $\bar{g}=(0.7, 0.3; 0.28, 0.7).$
\begin{figure}[t]
\begin{center}
\includegraphics[height=2.0in]{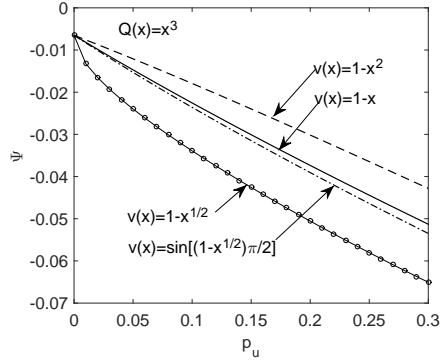}
\caption{Case I: The variations of $\Psi$ versus $p_{u}$ for the selected fear functions under the regret function $Q(x)=x^{3}.$} \label{fig1}
\end{center}
\end{figure}
First, let us consider that there is an unknown risk with the probability $p_{u}$ for the surgery treatment. The probabilities of the known outcomes are supposed to be revised as $0.6-0.5p_{u}$ and $0.4-0.5p_{u},$ respectively.
For a particular case of $p_{u}=0.1,$ the decision matrix should be adjusted and shown in Table \ref{tab:d4} by considering the independence of $\bar{f}$ and $\bar{g}.$
If we further assume that the fear function is linear with $v(x)=1-x,$ then the value of the function in (\ref{req11a}) can be computed as $\Psi=-0.0225<-0.0065.$ This means that the unknown risk enhances the choice of $\bar{f}\prec \bar{g}.$ This observation is in agreement with the common sense of people.
\begin{table}[ptb]
\caption{The decision matrix for the medical decision problem by considering an unknown outcome.}
\resizebox{\textwidth}{!}
	{
\begin{tabular}
[c]{c|cccccc}
 \hline
States &$s_{1}$&$s_{2}$&$s_{3}$&$s_{4}$&$s_{5}$&$s_{6}$\\
\hline
Probabilities&0.165 &0.385&0.105&0.245&0.03&0.07\\
\hline
Utilities for surgery &$0.5v(0.1)$&$0.5v(0.1)$&$0.27v(0.1)$&$0.27v(0.1)$&0&0\\
\hline
Utilities for radiotherapy &0.7&0.28&0.7&0.28&0.7&0.28\\
\hline
\end{tabular}\label{tab:d4}
}
\end{table}
Moreover, the fear function $v(x)$ may be nonlinear such as
$$
v(x)=1-x^{a},\quad v(x)=\sin\left[\frac{\pi}{2}(1-x^{a})\right],\quad a>0.
$$
under the conditions of $v(0)=1$ and $v(1)=0.$ Under the regret function $Q(x)=x^{3},$ the variations of $\Psi$ versus $p_{u}$ are shown in Figure \ref{fig1} by selecting several specific functions. It is seen from Figure \ref{fig1} that with the increasing of $p_{u},$ the values of $\Psi$ are decreasing. This means that the lager the probability of the unknown risk for surgery, the more the likelihood of choosing radiotherapy is. The above phenomenon corresponds to the experimental evidence of percentage increasing shown in Table \ref{lab2a}.

\item Case II: $\bar{f}=(0.5, 0.60; 0.27, 0.40)$ and $$\bar{g}=(0.7, 0.3-0.5p_{u}; 0.28, 7-0.5p_{u}; \text{unknown outcome}, p_{u}).$$
\begin{figure}[t]
\begin{center}
\includegraphics[height=2.0in]{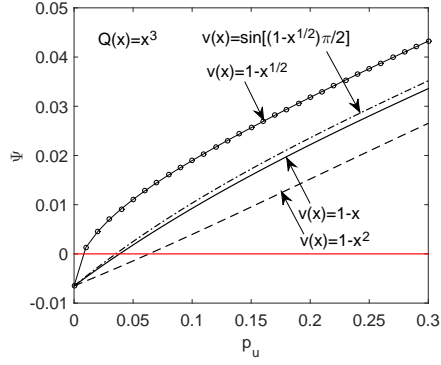}
\caption{Case II: The variations of $\Psi$ versus $p_{u}$ for the selected fear functions under the regret function $Q(x)=x^{3}.$} \label{fig2}
\end{center}
\end{figure}
Second, it is supposed that there is an unknown consequence with the probability $p_{u}$ for the radiotherapy treatment. The probabilities of the known outcomes are adjusted as $0.3-0.5p_{u}$ and $0.7-0.5p_{u},$ respectively. Similar to Table \ref{tab:d4}, the corresponding decision matrix is given in Table \ref{tab:d5} for $p_{u}=0.1.$
If choosing $v(x)=1-x,$ we have $\Psi=0.0092>0,$ meaning that $\bar{f}\succ \bar{g}.$
In other words, the existence of the unknown outcome could make the preference of the DM be reversal. Furthermore, it is interesting to consider the effects of different fear functions on the choice. Under the regret function $Q(x)=x^{3},$ the variations of $\Psi$ versus $p_{u}$ are drawn in Figure \ref{fig2} for the selected fear functions. With the increasing of $p_{u},$ the values of $\Psi$ are increasing. The observation shows that the percentage decreasing phenomenon occurs similar to that in the experimental evidence in Table \ref{lab2a}. In addition, it is found from Figure \ref{fig2} that there a critical point for preference reversal of the DM under each fear function. This implies that the preference reversal is dependent on the combined effect of the psychology of fearing the unknown and the probability of the occurrence of the unknown.

\begin{table}[ptb]
\caption{The decision matrix for the medical decision problem by considering an unknown outcome.}
\resizebox{\textwidth}{!}
	{
\begin{tabular}
[c]{c|cccccc}
 \hline
States &$s_{1}$&$s_{2}$&$s_{3}$&$s_{4}$&$s_{5}$&$s_{6}$\\
\hline
Probabilities&0.15 &0.39&0.10&0.26&0.06&0.04\\
\hline
Utilities for surgery &$0.5$&$0.5$&$0.27$&$0.27$&0.5&0.27\\
\hline
Utilities for radiotherapy &$0.7v(0.1)$&$0.28v(0.1)$&$0.7v(0.1)$&$0.28v(0.1)$&0&0\\
\hline
\end{tabular}\label{tab:d5}
}
\end{table}

\item Case III: We consider that there are unknown outcomes for the two treatments such as
$$\bar{f}=(0.5, 0.6-0.5p_{fu}; 0.27, 0.4-0.5p_{fu}; \text{unknown outcome}, p_{fu}),$$ and
$$\bar{g}=(0.7, 0.3-0.5p_{gu}; 0.28, 0.7-0.5p_{gu}; \text{unknown outcome}, p_{gu}).$$
\begin{table}[ptb]
\caption{The decision matrix for the medical decision problem by considering an unknown outcome.}
\resizebox{\textwidth}{!}
	{
\begin{tabular}
[c]{c|ccccccccc}
 \hline
States &$s_{1}$&$s_{2}$&$s_{3}$&$s_{4}$&$s_{5}$&$s_{6}$&$s_{7}$&$s_{8}$&$s_{9}$\\
\hline
Probabilities&0.1375 &0.3575&0.0875&0.2275&0.025&0.065&0.055&0.035&0.01\\
\hline
Utilities for surgery &$0.5v(0.1)$&$0.5v(0.1)$&$0.27v(0.1)$&$0.27v(0.1)$&0&0&$0.5v(0.1)$&$0.27v(0.1)$&0\\
\hline
Utilities for radiotherapy &$0.7v(0.1)$&$0.28v(0.1)$&$0.7v(0.1)$&$0.28v(0.1)$&$0.7v(0.1)$&$0.28v(0.1)$&0&0&0\\
\hline
\end{tabular}}\label{tab:d6}
\end{table}
\begin{figure}[t]
\begin{center}
\includegraphics[height=2.0in]{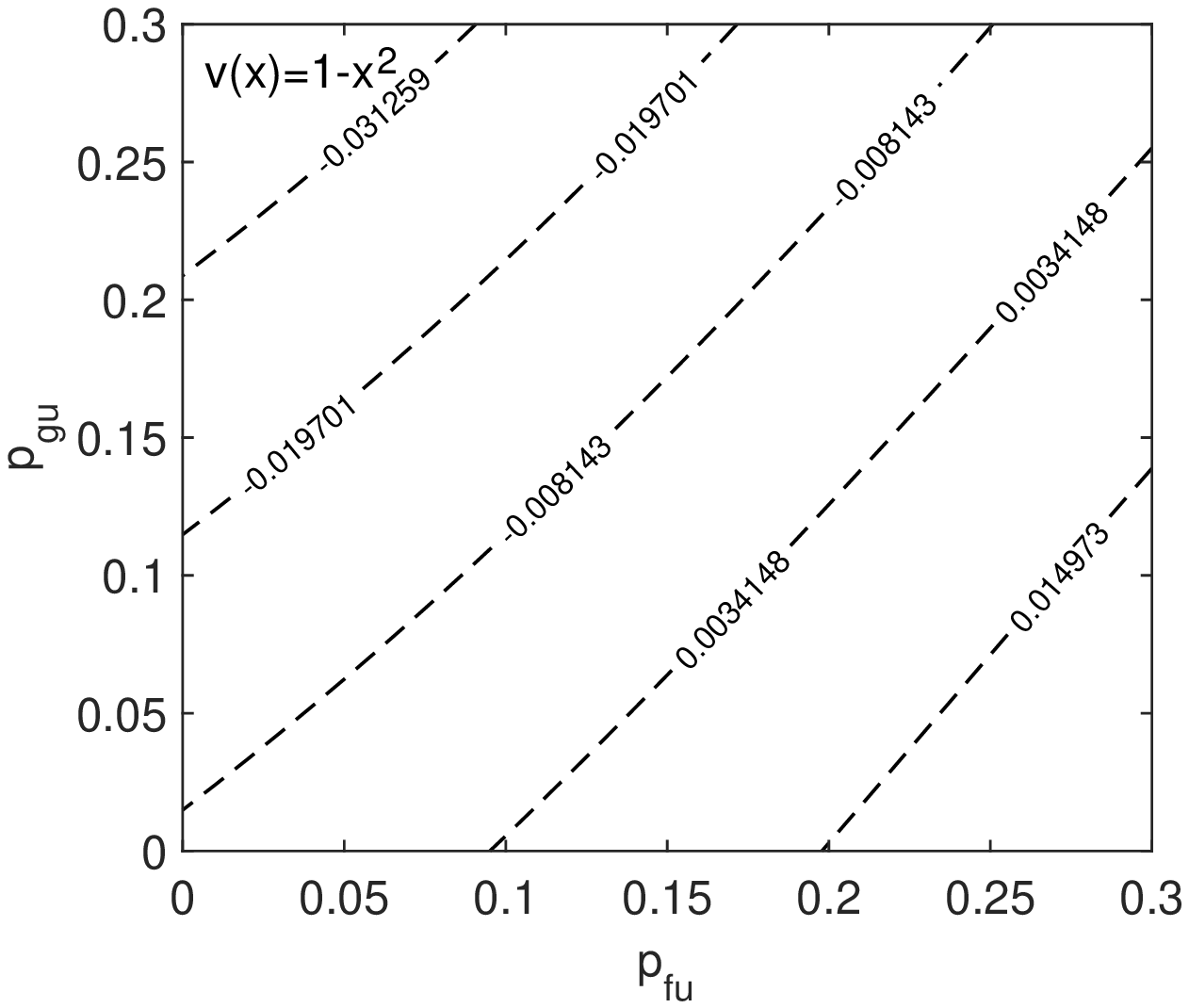}\\
\includegraphics[height=2.0in]{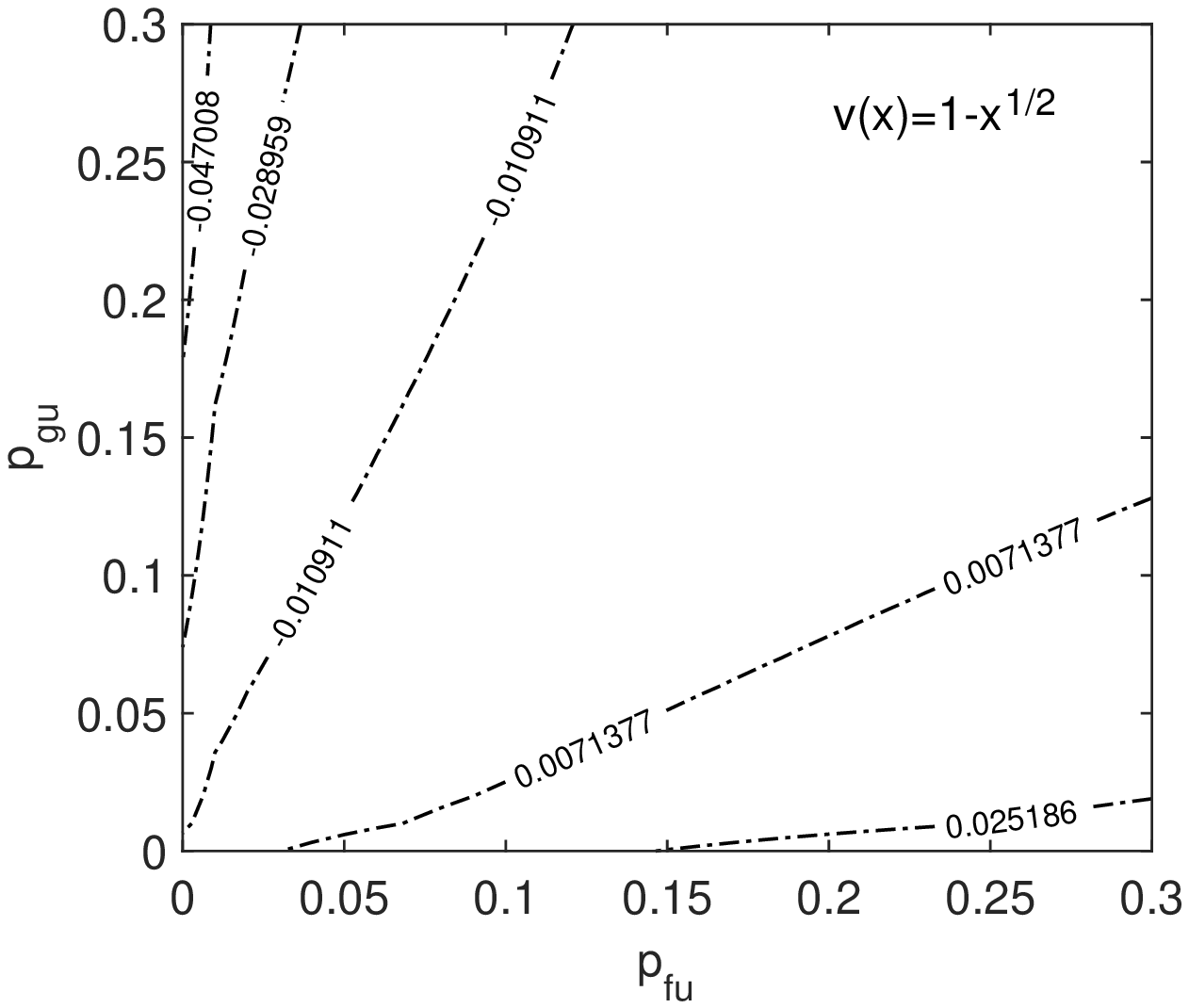}
\caption{Case III: The contours of $\Psi$ versus $p_{fu}$ and $p_{gu}$ for the selected fear functions under the regret function $Q(x)=x^{3}.$} \label{fig3}
\end{center}
\end{figure}
For example, if we choose $p_{fu}=p_{gu}=0.1,$ the decision matrix can be given in Table \ref{tab:d6}. After some computations, we arrive at $\Psi=-0.0049<0$ by using $v(x)=1-x,$ which yields $\bar{f}\prec \bar{g}.$ It is still interesting to consider the effects of the fear functions together with the probabilities $p_{fu}$ and $p_{gu}$ on the preference relation of the DM. For the sake of comparisons, we select $v(x)=1-x^{2}$ and $v(x)=1-x^{1/2}$ for numerical computations. The contour lines of $\Psi$ with respect to $p_{fu}$ and $p_{gu}$ are depicted in Figure \ref{fig3}. It is seen that the values of $\Psi$ are affected by the factors of $p_{fu},$ $p_{gu}$ and the fear function $v(x).$ This means that the preferences of the DM could be confused due to multiple unknowns.
\end{itemize}

Based on the above discussions, one can see that
the unknown outcomes may enhance, impede or reverse the existing preference relation. It is interesting to find that the theoretical analysis is in agreement with the experimental evidences.

\section{Conclusions}

The modified regret theory has been developed by considering the existence of an unknown outcome in a prospect. The experimental evidences have been revealed by analyzing the responses of 109 undergraduate students at Guangxi University to some hypothetical choice problems. The theoretical framework of the modified regret theory under fear of the unknown has been established and the facts derived from the developed theory have been given. Following the idea in \cite{Diecidue2017}, the utility-based behavioral foundation of the modified regret theory has been constructed. The medical decision problem has been reinvestigated by considering an unknown risk. It is found that the final choice of the DM is affected greatly by the fear functions and the probability of an unknown outcome. The observations show that the theoretical analysis is in accordance with the experimental evidences. The modified regret theory can be used to analyze the practical decision problems where an unknown outcome is existing.

\section*{Acknowledgements}

The work was supported by the National Natural Science Foundation of China (Nos. 71871072, 71571054), 2017 Guangxi high school innovation team and outstanding scholars plan, the Guangxi Natural Science Foundation for Distinguished Young Scholars (No. 2016GXNSFFA380004), and the Innovation Project of Guangxi Graduate Education (No. YCSW2021044).




\begin{thebibliography}{99}

\bibitem[Bell(1981)]{Bell1981}Bell, David E., 1981. Components of risk aversion. in J.P. Brans (ed.), Operational
Research'81, Amsterdam: North-Holland, pp. 371--378.

\bibitem[Bell(1982)]{Bell1982}Bell, David E., 1982. Regret in decision making under uncertainty. Oper. Res. 30 (5), 961--981.


\bibitem[Bleichrodt et al.(2010)]{Bleichrodt2010}Bleichrodt, Han, Cillo, Alessandra, Diecidue, Enrico, 2010. A quantitative measurement of regret theory. Management Sci. 56 (1), 161--175.

\bibitem[Bleichrodt and Wakker(2015)]{Bleichrodt2015}Bleichrodt, Han, Wakker, Peter P., 2015. Regret theory: a bold alternative to the alternatives. Econ. J. 125 (583), 493--532.


\bibitem[Bourgeois-Gironde(2010)]{Bour2010}Bourgeois-Gironde, Sacha, 2010. Regret and the rationality of choices. Philos. Trans. R. Soc. Lond. B, Biol. Sci. 365 (1538), 249--257.


\bibitem[Camille et al.(2004)]{Camille2004}Camille, Nathalie, Coricelli, Giorgio, Sallet, Jerome, Pradat-Diehl, Pascale, Duhamel, Jean-Ren\'{e}, Sirigu, Angela, 2004. The involvement of the orbitofrontal cortex in the experience of regret. Science 304 (5674), 1167--1170

\bibitem[Cao et al.(2011)]{Cao2011}Cao, H. Henry, Han, Bing, Hirshleifer, David, Zhang, Harold H., 2011. Fear of the unknown: Familiarity and economic
decisions. Rev. Finance 15, 173--206.

\bibitem[Carleton(2016)]{Carleton2016}Carleton, R. Nicholas, 2016. Into the unknown: A review and synthesis of contemporary models
involving uncertainty. J. Anxiety Disorders 39, 30--43.

\bibitem[Diecidue and Somasundaram(2017)]{Diecidue2017}Diecidue, Enrico, Somasundaram, Jeeva, 2017. Regret theory: A new foundation. J. Econ. Theory 172, 88--119.

\bibitem[Du and Budescu(2005)]{Du2005} Du, Ning, Budescu, David V., 2005. The effects of imprecise probabilities and
outcomes in evaluating investment options. Management Sci. 51(12), 1791--1803.

\bibitem[Fishburn(1970)]{Fishburn1970}Fishburn, Peter C., 1970. Utility theory for decision making. John Wiley \& Sons, INC.

\bibitem[Fishburn(1989)]{Fishburn1989}Fishburn, Peter C., 1989. Non-transitive measurable utility for decision under uncertainty. J. Math. Econ. 18 (2), 187--207.

\bibitem[Jones et al.(2005)]{Jones2005}Jones, Sarah, Statham, Helen, Solomou, Wendy, 2005. When expectant mothers know
their baby has a fetal abnormality: Exploring the crisis of motherhood through
qualitative data-mining. J. Social Work Res. Evalua. 6(2),
195--206.

\bibitem[Kahneman and Tversky(1979)]{Kahneman1979}Kahneman, Daniel, Tversky, Amos, 1979. Prospect theory: an analysis of decision under risk. Econometrica 47(2), 263--292.

\bibitem[K\"{o}bberling and Wakker(2004)]{Kobberling2004}K\"{o}bberling, Veronika, Wakker, Peter P., 2004. A simple tool for qualitatively testing, quantitatively measuring, and normatively justifying Savage's subjective expected utility. J. Risk Uncertain. 28 (2), 135--145.

\bibitem[Loomes and Sugden(1982)]{Loomes1982}Loomes, Graham, Sugden, Robert, 1982. Regret theory: an alternative theory of rational choice under uncertainty. Econ. J. 92 (368), 805--824.

\bibitem[Loomes and Sugden(1987)]{Loomes1987}Loomes, Graham, Sugden, Robert, 1987. Some implications of a more general form of regret theory. J. Econ. Theory 41 (2), 270--287.

\bibitem[Quiggin(1994)]{Quiggin1994}Quiggin, John, 1994. Regret theory with general choice sets. J. Risk Uncertain. 8 (2), 153--165.

\bibitem[von Neumann and Morgenstern(1953)]{Neumann1953}von Neumann, John, Morgenstern, Oskar, 1953. The theory of games and
economic behaviour (Third Version). Princeton University Press.

\bibitem[Wakker(1984)]{Wakker1984}Wakker, Peter, 1984. Cardinal coordinate independence for expected utility. J. Math. Psychol. 28 (1), 110--117.

\bibitem[Wakker(2010)]{Wakker2010}Wakker, Peter, 2010. Prospect Theory for Risk and Ambiguity. Cambridge University Press.

\end{thebibliography}
\end{document}